\documentclass[11pt]{article}
\usepackage[margin=1in]{geometry}
\usepackage{amsfonts}
\usepackage{color}
\usepackage{amsthm,amsmath}
\usepackage{algorithm,algorithmic}
\usepackage{color}
\usepackage{graphicx}
\usepackage{url}
\usepackage{wrapfig}
\usepackage{latexsym}
\usepackage{hyperref}
\usepackage{verbatim}
\usepackage[justification=centering]{caption}
\usepackage{enumitem}
\allowdisplaybreaks[1]
\usepackage[numbers]{natbib}
\usepackage{setspace}

\newtheorem{theorem}{\hskip\parindent\bf{Theorem}}

\newtheorem{lemma}{\hskip\parindent\bf{Lemma}}
\newtheorem{proposition}{\hskip\parindent\bf{Proposition}}

\newtheorem{definition}{\hskip\parindent\bf{Definition}}

\def\Pr{\qopname\relax n{\mathbf{Pr}}}

\renewcommand{\bar}{\overline}

\def\min{\qopname\relax n{min}}
\def\max{\qopname\relax n{max}}

\def\Pr{\qopname\relax n{\mathbf{Pr}}}

\newcommand{\RR}{\mathbb{R}}

\newcommand{\mini}[1]{\mbox{minimize} & {#1} &\\}
\newcommand{\maxi}[1]{\mbox{maximize} & {#1 } & \\}

\newcommand{\st}{\mbox{subject to} }
\newcommand{\con}[1]{&#1 & \\}
\newcommand{\qcon}[2]{&#1, & \mbox{for } #2.  \\}
\newenvironment{lp}{\begin{equation}  \begin{array}{lll}}{\end{array}\end{equation}}
\newenvironment{lp*}{\begin{equation*}  \begin{array}{lll}}{\end{array}\end{equation*}}

\date{}

\begin{document}
\title{Security Games with Information Leakage: Modeling and Computation}
\author{
	\begin{tabular}{ccccc}
		\centering
		 Haifeng Xu \thanks{University of Southern California, Los Angeles, USA; \{haifengx,aruneshs,shaddin,tambe\}@usc.edu.  Shaddin Dughmi is supported in part by  NSF CAREER Award CCF-1350900. Haifeng Xu is supported by NSF grant CCF-1350900, MURI grant W911NF-11-1-0332 and US-Naval Research grant Z14-12072.} & \quad \quad &   Albert X. Jiang \thanks{Trinity University, San Antonio, USA; xjiang@trinity.edu.} & \quad \quad  &  Arunesh Sinha \footnotemark[1]\tabularnewline
		& & & & \tabularnewline
		Zinovi Rabinovich \thanks{Independent Researcher, Israel; zr@zinovi.net. }& & Shaddin Dughmi \footnotemark[1]&  & Milind Tambe \footnotemark[1]  \tabularnewline
	\end{tabular}	
	}
	
\maketitle

\begin{abstract}
Most models of Stackelberg security games assume that the attacker only knows the defender's mixed strategy, but is not able to observe (even partially) the instantiated pure strategy. Such partial observation of the deployed pure strategy  -- an issue we refer to as {\it information leakage} -- is a significant concern in practical applications. While previous research on patrolling games has considered the attacker's real-time surveillance, our settings, therefore models and techniques, are fundamentally different. More specifically, after describing the information leakage model, we start with an LP formulation to compute the defender's optimal strategy in the presence of  leakage. Perhaps surprisingly, we show that a key subproblem to solve this LP (more precisely, the defender oracle) is NP-hard {\it even} for the simplest of security game models. We then approach the problem from three possible directions: efficient algorithms for restricted cases, approximation algorithms, and heuristic algorithms for sampling that improves upon the status quo. Our experiments confirm the necessity of handling information leakage and the advantage of our algorithms.



\end{abstract}

\newpage

\section{Introduction}
Stackelberg security games played between a defender (leader) and an attacker (follower) have been widely studied in the past few years \cite{korzhyk_etal_2010,Joshua11,letchford_etal_2011,Tambe2011}. Most models, in particular, including all the deployed security systems in \cite{Tambe2011}, assume that the attacker is not able to observe (even partially) the defender's instantiated pure strategy (i.e., which targets are being protected), thus he makes decisions based only on his knowledge of the defender's mixed strategy. This fails to capture the attacker's real-time surveillance, by which he may {\it partially} observe the deployed pure strategy. For example, the attacker may observe the protection status of a certain target while approaching for an attack; or in some security domains information regarding the protection status of certain targets may leak to the attacker due to real-time surveillance or even an insider threat; further, well-prepared attackers may approach certain adversarially chosen target to collect information before committing an attack.

Unfortunately, this problem-- an issue we refer to as {\it information leakage} -- has not received much attention in Stackelberg security games. In the literature of patrolling games, attackers' real-time surveillance is indeed considered \cite{noa08a,noa08b,Gatti1,Gatti2,Bosansky,Vorobeychik14a}. However, all these papers study settings of patrols carried out over space and time, i.e., the defender follows a schedule of visits to multiple  targets over time. In addition, they assume that it takes time for the attacker to execute an attack, during which the defender can interrupt the attacker by visiting the attacked target.
Therefore, even if the attacker can fully observe the current position of the defender (in essence, status of {\it all} targets), he may not have enough time to complete an attack on a target before being interrupted by the defender. The main challenge there is to create  patrolling schedules with the smallest possible time between any two target visits.  In contrast, we consider information leakage in standard security game models, where the attack is \emph{instantaneous} and cannot be interrupted by the defender's resource re-allocation. Furthermore, as may be more realistic in our settings, we assume that information is leaked from a limited number of targets. 
As a result, our setting necessitates novel models and techniques. We also provide efficient algorithms with complexity analysis.


This paper considers the design of optimal defender strategy in the presence of {\it partial} information leakage.
Considering that real-time surveillance is costly in practice, we explicitly assume that information leaks from {\it only one} target, though our model and algorithms can be generalized.
We start from the basic security game model where the defender allocates $k$ resources to protect $n$ targets without any scheduling constraint. Such models have applications in real security systems like ARMOR for LAX airport and GUARDS  for airports in general~\cite{Tambe2011}. 
We first show via a concrete example in Section~\ref{sec:model} how ignoring information leakage can lead to significant utility loss. This motivates our design of optimal defending strategy given the possibility of information leakage. We start with a linear program formulation.
However, surprisingly, we show that
it is difficult to solve the LP {\it even} for this basic case, whereas the optimal mixed strategy without leakage can be computed easily. In particular, we show that the defender oracle, a key subproblem used in the column generation technique employed for most security games, is NP-hard. 
This shows the intrinsic difficulty of handling information leakage. We then approach the problem from three directions: efficient algorithms for special cases, approximation algorithms and heuristic algorithms for sampling that improves upon the status quo. Our experiments support our hypothesis that ignoring information leakage can result in significant loss of utility for the defender, and demonstrates the value of our algorithms.

\section{Model of Information Leakage}\label{sec:model}
Consider a standard zero-sum Stackelberg security game with a defender and an attacker. The defender allocates $k$ security resources to protect $n$ targets, which are denoted by the set  $[n]=\{1,2,...,n\}$. In this paper we consider the case where the security resources do {\it not} have scheduling constraints. That is, the defender's pure strategy is to protect any subset of $[n]$ of size at most $k$. For any $i \in [n]$, let $r_i$  be the reward [$c_i$ be the cost] of the defender when the attacked target $i$ is protected [unprotected]. We consider zero-sum games, therefore the attacker's utility is the negation of the defender's utility. Let $s$ denote a pure strategy and $S$ be the set of all possible pure strategies. With some abuse of notation, we sometimes regard $s$ as a \emph{subset} of $[n]$ denoting the protected targets; and sometimes view it as an $n$-dimensional $0-1$ \emph{vector} with $k$ $1$'s specifying the protected targets. The intended interpretation should be clear from context.   The {\it support} of a mixed strategy is defined to be the set of pure strategies with non-zero probabilities. 
Without information leakage, the problem of computing the defender's optimal mixed strategy can be compactly formulated as linear program \eqref{lp:approxOPT} with each variable $x_i$ as the marginal probability of covering target $i$. The resulting marginal vector $\vec{x}$  is a convex combination of the indicator vectors of pure strategies, and a mixed strategy with small support can be efficiently sampled, e.g., by Comb Sampling \cite{Tsai10a}.
\begin{lp} \label{lp:approxOPT} 
	\maxi{u} 
	\st
	\qcon{u\leq r_i x_{i} +c_i (1-x_{i})}{i\in[n]}
	\con{\sum_{i\in[n]} x_{i} \leq k}
	\qcon{0 \leq x_{i} \leq 1}{i \in [n]} 
\end{lp}

Building on this basic security game, our model goes one step further and considers the possibility that the protection status of one  target leaks to the attacker. Here, by ``protection status" we mean whether this target is protected or not in an {\it instantiation} of the mixed strategy. We consider two related models of information leakage:
\begin{enumerate}
	\item {\bf PR}obabilistic {\bf I}nformation {\bf I}eakage (PRIL): with probability $p_i(\geq 0)$ a \emph{single} target $i$ leaks information; and  with probability $p_0 =1-\sum_{i=1}^n p_i$  no targets leak information.  So we have  $\vec{p}=(p_0,p_1,...,p_n)\in \Delta_{n+1}$ where $\Delta_{n+1}$ is the $(n+1)$-dimensional simplex. In practice, $\vec{p}$ is usually given by domain experts and may be determined by the nature or property of targets. 
	\item {\bf AD}versarial {\bf I}nformation {\bf L}eakage (ADIL): with probability $1-p_0$, one {\em adversarially} chosen target leaks  information. This model captures the case where the attacker will strategically choose a target for surveillance and with certain probability he succeeds in observing the protection status of the surveyed target.   
\end{enumerate}   


Given either model -- PRIL with any $\vec{p}\in \Delta_{n+1}$ or ADIL -- we are interested in computing the optimal defender patrolling strategy. The first question to ask is: why does the issue of information leakage matter and how does it affect the computation of the optimal defender strategy? To answer this question we employ a concrete example.

Consider a zero-sum security game with $4$ targets and $2$ resources. The profiles of reward $r_i$ [cost $c_i$] is $\vec{r}=(1,1,2,2)$ [$\vec{c}=(-2,-2,-1,-1)$], where the coordinates are indexed by target ids. 
If there is no information leakage, it is easy to see that the optimal marginal coverage is $\vec{x}=(\frac{2}{3},\frac{2}{3},\frac{1}{3},\frac{1}{3})$. The attacker will attack an arbitrary target, resulting in a defender utility of $0$. Now, let us consider a simple case of information leakage. Assume the attacker observes whether target $1$ is protected or not in any instantiation of the mixed strategy, i.e., $p_1=1$. As we will argue, how the marginal probability $\vec{x}$ is implemented would matter now. One way to implement $\vec{x}$ is to protect target $\{1,2\}$  with probability $\frac{2}{3}$ and protect $\{3,4\}$  with probability $\frac{1}{3}$. However, this implementation is  ``fragile" in the presence of the above information leakage. In particular, if the attacker observes that target $1$ is protected (which occurs with probability $\frac{2}{3}$), he infers that the defender is protecting target $\{1,2\}$ and will attack $3$ or $4$, resulting in a defender utility of $-1$; if target $1$ is not protected, the attacker will just attack,  resulting in a defender utility of $-2$. Therefore, the defender gets expected utility $-\frac{4}{3}$. 

Now consider another way to implement the {\it same} marginal $\vec{x}$ by the following mixed strategy:

\begin{table}[H]
\centering{}
\begin{tabular}{cccccc}
	\hline 
	$\{1,2\}$ & $\{1,3\}$ & $\{1,4\}$ & $\{2,3\}$ & $\{2,4\}$ & $\{3,4\}$\tabularnewline 
	\hline 
	$10/27$ & $4/27$ & $4/27$ & $4/27$ & $4/27$ & $1/27$\tabularnewline
	\hline 
\end{tabular}
\end{table}
If the attacker observes that target $1$ is protected (which occurs with probability $\frac{2}{3}$), then he infers that target $2$ is protected with probability $\frac{\frac{10}{27}}{\frac{10}{27}+\frac{4}{27}+\frac{4}{27}}=\frac{5}{9}$, and target $3,4$ are both protected with probability $\frac{2}{9}$. Some calculation shows that the attacker will have the same utility $\frac{1}{3}$ on target $2,3,4$ and thus will choose an arbitrary one to attack, resulting in a defender utility of $-\frac{1}{3}$. On the other hand, if target $1$ is observed to be unprotected, the defender gets utility $-2$. In expectation, the defender gets utility $\frac{2}{3}\times (-\frac{1}{3})+\frac{1}{3}\times(-2)=-\frac{8}{9}$. 

As seen above, though implementing the same marginals, the latter mixed strategy achieves better defender utility than the former one in the presence of information leakage. However, is it optimal? It turns out that the following mixed strategy achieves an even better defender utility of $-\frac{1}{3}$, which can be proved to be optimal: 
protect $\{1,2\}$ with probability $\frac{5}{9}$, $\{1,3\}$ with probability $\frac{2}{9}$ and $\{1,4\}$ with probability $\frac{2}{9}$.  

This example shows that compact representation by marginal coverage probabilities is not sufficient for computing the optimal defending strategy assuming information leakage. This naturally raises new computational challenges: how can we formulate the defender's optimization problem and compute the optimal solution? Is there still a compact formulation or is it necessary to enumerate all the  exponentially many pure strategies? What is the computational complexity of this problem? We answer these questions in the next section.

\section{Computing Optimal Defender Strategy}
We will focus on the derivation of the PRIL model. The formulation for the ADIL model is provided at the end of this section since it admits a similar derivation.  Fixing the defender's mixed strategy, let $t_i$ ($\neg t_i$)  denote the event that target $i$ is { \it protected} ({\it unprotected}). For the PRIL model, the defender's utility equals
$$
\begin{array}{c}
DefU = p_0 u+\sum_{i=1}^n p_i (u_i+v_i)
\end{array}
$$
where $u = \min_j \left[r_j \Pr (t_j)+c_j \Pr (\neg  t_j) \right] $ is the defender's utility when there is no information leakage; and 
$$
\begin{array}{c c l}
u_i &=& \Pr (t_i) \times \min_j \left[ r_j \Pr(t_j|t_i)+c_j \Pr(\neg  t_j|t_i) \right] \\
&=&\min_j \left[ r_j \Pr (t_j,t_i)+c_j \Pr (\neg  t_j,t_i) \right]
\end{array}
$$
is the defender's utility when target $i$ leaks out its protection status as $t_i$ (i.e., protected) multiplied by probability $\Pr (t_i)$. Similarly 
$$
\begin{array}{c c l}
v_i &=& \min_j \left[ r_j \Pr (t_j,\neg  t_i)+c_j \Pr (\neg  t_j, \neg  t_i) \right]
\end{array}
$$
is the defender's expected utility multiplied by probability $\Pr (\neg t_i)$  when target $i$ leaks status $\neg t_i$ (i.e., unprotected) 

Define variables $x_{ij}=\Pr (t_i,t_j)$ (setting $x_{ii}=\Pr (t_i)$). Using the fact that $\Pr (t_i,\neg t_j)=x_{ii} - x_{ij}$ and $\Pr (\neg t_i,\neg t_j) = 1-x_{ii}- x_{jj}+x_{ij}$, we obtain the following linear program which computes the defender's optimal patrolling strategy:
\begin{lp}\label{lp:naiveLP}
	\maxi{p_0 u+\sum_{i=1}^n p_i(u_i +v_i)}
	\st
	\qcon{ u\leq r_j x_{jj} +c_j (1-x_{jj}) }{j \in [n]}
	\qcon{u_{i}\leq r_j x_{ij}+c_j (x_{ii}-x_{ij}) }{ i,j \in [n]}
	\qcon{v_{i}\leq r_j (x_{jj}-x_{ij})+ c_j (1-x_{ii}- x_{jj}+x_{ij}) }{ i,j \in [n]}
	\qcon{ x_{ij}=\sum_{s: i,j \in s} \theta_s }{ i,j \in [n]}
	\con{ \sum_{s\in S} \theta_{s} = 1}
	\qcon{\theta_s \geq 0}{  s \in S}
\end{lp}  
where $u,u_i,v_i,x_{ij},\theta_s$  are variables; $s$ denotes a pure strategy and the sum condition ``$s:i,j \in s$" means summing over all the pure strategies that protect both targets $i$ and $j$ (or $i$ if $i=j$); $\theta_s$ denotes the probability of choosing strategy $s$.

Unfortunately, LP~\eqref{lp:naiveLP} suffers from an exponential explosion of variables, specifically, $\theta_s$. From the sake of computational efficiency, one natural idea is to find a compact representation of the defender's mixed strategy. As suggested by LP~\eqref{lp:naiveLP}, the variables $x_{ij}$, indicating the probability that targets $i,j$ are both protected, are sufficient to describe the defender's objective and the attacker's incentive constraints. 

Let us call variables $x_{ij}$ the { \it pair-wise marginals}  and think of them as a matrix $X\in \mathbb{R}^{n\times n}$, i.e., the $i$'th row and $j$'th column of $X$ is $x_{ij}$ (not to be confused with the {\it marginals} $\vec{x}$).  We say $X$ is {\it feasible} if there exists a mixed strategy, i.e., a distribution over pure strategies, that achieves the pair-wise marginals $X$. Clearly, not all $X\in \mathbb{R}^{n\times n}$ are feasible. Let $\mathcal{P}(n,k)\in \mathbb{R}^{n\times n}$ be the set of all {\it feasible} $X$. The following lemma shows a structural property of $\mathcal{P}(n,k)$.
\begin{lemma}\label{lem:polytope}
	$\mathcal{P}(n,k)$ is a polytope and any $X\in \mathcal{P}(n,k)$ is a \emph{symmetric positive semi-definite} (PSD) matrix.
\end{lemma}
\begin{proof}
	Notice that $X$ is feasible if and only there exists $\theta_s$ for any pure strategy $s$ such that the following linear constraints hold:
	\begin{lp}
		\qcon{ x_{ij}=\sum_{s: i,j \in s} \theta_s }{ i,j \in [n]}
		\con{ \sum_{s \in S} \theta_{s} = 1}
		\qcon{\theta_s \geq 0}{ s \in S}
	\end{lp}
	
	These constraints define a polytope for variables $(X,\vec{\theta})$, therefore its projection to the lower dimension $X$, which is precisely $\mathcal{P}(n,k)$, is also a polytope.
	
	To prove $X\in \mathcal{P}(n,k)$ is PSD, we first observe that any vertex of $\mathcal{P}(n,k)$, characterizing a pure strategy, is PSD. In fact, let $s \in \{ 0,1 \}^n$ be any pure strategy, then the pair-wise marginal w.r.t. $s$ is $X_s = s s^T$, which is PSD. Therefore, any  $X \in \mathcal{P}$, which is a convex combination of its vertices, is also PSD.
\end{proof}

\begin{lp}\label{lp:compactLP}
	\maxi{p_0 u+\sum_{i=1}^n p_i(u_i +v_i)}
	\st
	\qcon{ u\leq r_j x_{jj} +c_j (1-x_{jj}) }{j \in [n]}
	\qcon{u_{i}\leq r_j x_{ij}+c_j (x_{ii}-x_{ij}) }{ i,j \in [n]}
	\qcon{v_{i}\leq r_j (x_{jj}-x_{ij})+ c_j (1-x_{ii}- x_{jj}+x_{ij}) }{ i,j \in [n]}
	\con{ X \in \mathcal{P}(n,k) }
\end{lp} 

With Lemma~\ref{lem:polytope}, we may re-write LP~\eqref{lp:naiveLP} compactly as LP~\eqref{lp:compactLP} with variables $u$, $u_i$, $v_i$ and $X$. Therefore,  we would be able to compute the optimal strategy efficiently in polynomial time if the constraints determining the polytope $\mathcal{P}(n,k)$ were only polynomially many -- recall that this is the approach we took with LP~\eqref{lp:approxOPT} in the case of no information leakage. However, perhaps surprisingly, the problem turns out to be  much harder in the presence of leakage.
\begin{lemma} \label{lem:optNP-hard}
	Optimizing over $\mathcal{P}(n,k)$ is NP-hard.
\end{lemma}
\begin{proof}
	We prove by reduction from the densest $k$-subgraph problem. Given any graph instance $G=(V,E)$,
	let $A$ be the adjacency matrix of $G$. Consider the following
	linear program:
	\begin{lp} \label{lp:densest}
		\maxi{\sum_{i,j\in [n]} A_{ij}x_{ij}} 
		\st
		\con{ X\in \mathcal{P}(n,k).}
	\end{lp}
	This linear program must have a {\it vertex} optimal solution $X^*$ which satisfies $X^* = s s^T$ for some pure strategy $s\in \{0,1\}^n$. Therefore, the linear objective satisfies
	 \begin{equation*}
	\sum_{i,j\in [n]} A_{ij}x_{ij} = tr(A X^*) = tr(A \times ss^T) = tr(s^T A s) = s^T A s.
	\end{equation*}
	Notice that $s^T A s/2k$ equals the density of a subgraph of $G$ with $k$ nodes indicated by $s$. Since $X^*$ is the optimal solution to LP~\eqref{lp:densest}, it also maximizes the density $s^T A s/2k$ over all subgraphs with $k$ nodes. In other words, the ability of optimizing  LP~\eqref{lp:densest} implies the ability of computing the densest $k$-subgraph, which is NP-hard. Therefore, optimizing over  $\mathcal{P}(n,k)$ is NP-hard. 
\end{proof}

 Lemma~\ref{lem:optNP-hard} suggests that there is no hope of finding polynomially many linear constraints which determine $\mathcal{P}(n,k)$ or, more generally, an efficient separation oracle for $\mathcal{P}(n,k)$, assuming $P\not=NP$. In fact,  $\mathcal{P}(n, k)$ is closely related to a fundamental geometric object, known as the {\it correlation polytope}, which has applications in quantum mechanics, statistics, machine learning and combinatorial problems. We show a connection between $\mathcal{P}(n,k)$ and the correlation polytope in Appendix B. For further information, we refer the reader to \cite{Corpoly}.

Another approach for computing the optimal defender strategy is to use the technique of column generation, which is a master/slave decomposition of an optimization problem. The essential part of this approach is the slave problem, which is also called  the ``defender best response oracle" or ``defender oracle" for short \cite{Jain10}. We defer the derivation of the defender oracle to Appendix A, while only mention that a similar reduction as in the proof of Lemma~\ref{lem:optNP-hard} also implies the follows.
\begin{lemma}
	The defender oracle is  NP-hard.
\end{lemma} 


By now, we have shown the evidence of the difficulty of solving LP~\eqref{lp:naiveLP} using either marginals or the technique of column generation. For the ADIL model, a similar argument yields that the following LP formulation computes the optimal defender strategy. It is easy to check that it shares the same marginals and defender oracle as the PRIL model.
\begin{lp}\label{lp:reducedADLP}
	\maxi{p_0 u+ (1-p_0)w}
	\st
	\qcon{ u\leq r_j x_{jj} +c_j (1-x_{jj}) }{j \in [n]}
	\qcon{u_{i}\leq r_j x_{ij}+c_j (x_{ii}-x_{ij}) }{ i,j \in [n]}
	\qcon{v_{i}\leq r_j (x_{jj}-x_{ij})+ c_j (1-x_{ii}- x_{jj}+x_{ij}) }{ i,j \in [n]}
	\qcon{w \leq u_i + v_i}{i\in[n]}
	\con{ X \in \mathcal{P}(n,k) }
\end{lp}
where variable $w$ is the defender's expected utility when an adversarially chosen target is observed by the attacker.

\subsection{Leakage from Small Support of Targets}\label{subsec:smallSupport}
Despite the hardness results for the general case, we show that the defender oracle admits a polynomial time algorithm if  information only  leaks from a small subset of targets; we call this set the {\it leakage support}.  By re-ordering the targets, we may assume without loss of generality that only the first $m$ targets, denoted by set $[m]$, could possibly leak information in both the PRIL and ADIL model. For the PRIL model, this means $p_i=0 $ for any $i>m$ and for the ADIL model, this means the attacker only chooses a target in $[m]$ for surveillance.

Why does this make the problem tractable? Intuitively the reason is as follows: when information leaks from a small set of targets, we only need to consider the correlations between these leaking targets and others, which is a much smaller set of variables than in LP~\eqref{lp:naiveLP} or \eqref{lp:reducedADLP}. Restricted to a leakage support of size $m$, the defender oracle is the following problem (See Appendix A for the derivation).  Let $A$ be a {\it symmetric} matrix of the following block form 
\begin{equation}\label{eqn:A}
	A: \, \left[\begin{array}{cc}
		A_{mm} & A_{mm'}\\
		A_{m'm} & A_{m'm'}
	\end{array}\right]
\end{equation}
where $m' = n-m$; $A_{mm'}\in \mathbb{R}^{m \times m'}$  for any integers $m,m'$ is a sub-matrix  and, crucially, $A_{m'm'}$ is a {\it diagonal} matrix.
Given $A$ of form \eqref{eqn:A}, find a pure strategy $s$ such that $s^TA s$ is maximized. That is, the defender oracle identifies the size-$k$ principle submatrix with maximum entry sum for any $A$ of form ~\eqref{eqn:A}. Note that $m=n$ in general case. 

Before detailing the algorithm, we first describe some notation. Let $A[i,:]$ be the $i$'th row of matrix $A$ and $diag(A)$ be the vector consisting of the diagonal entries of $A$. For any subset $C_1,C_2$ of $[n]$, let $A_{C_1,C_2}$ be  the submatrix of $A$ consisting of rows in $C_1$ and columns in $C_2$, and   $sum(A_{C_1,C_2})=\sum_{i\in C_1, j \in C_2} A_{ij}$  be the entry sum of $A_{C_1,C_2}$. The following lemma shows that Algorithm~\ref{alg:DefenderOracle} solves the defender oracle. Our main insight is that for a pure strategy $s$ to be optimal, once the set $C=s \cap [m]$ is decided, its complement $\bar{C} = s\setminus C$ can be explicitly identified, therefore we can simply brute-force search to find the best $C\subseteq [m]$. Lemma~\ref{lem:polyOracle} provides the algorithm guarantee, which then yields the polynomial solvability for the case of small $m$ (Theorem~\ref{thm:smallSupport}).
\begin{lemma}\label{lem:polyOracle}
	Let $m$ be the size of the leakage support. Algorithm~\ref{alg:DefenderOracle} solves the defender oracle and runs in  $ poly(n,k,2^m)$ time. In particular, the defender oracle admits a $poly(n,k)$ time algorithm if $m$ is a constant. 
\end{lemma}

\begin{proof}
	First, it is easy to see that Algorithm 1 runs in $poly(2^m ,n,k)$ time since the for-loop is executed at most $2^m$ times. We show that it solves the defender oracle problem.
	
	Let $s$ denote the indices of the principle submatrix of $A$ with maximum entry sum. Notice that $s$ can also be viewed as a pure strategy. Let $C = s \cap [m]$ and $\bar{C} = s \setminus C$. We claim that, given $C$, $\bar{C}$ must be the set of indices of the largest $k-|C|$ values from the set $\{ v_{m+1},...,v_{n} \}$, where $\vec{v}$ is defined as $\vec{v} = 2 \sum_{i\in C} A[i,:]+ diag(A)$. In other words, if we know $C$, the set $\bar{C}$ can be easily identified. To prove the claim, we re-write the $sum(A_{s,s})$ as follows:
	\begin{eqnarray*}
		& & sum(A_{s,s})\\
		&=&  sum(A_{C,C}) + 2sum(A_{C,\bar{C}}) + sum(A_{\bar{C},\bar{C}}) \\
		&=& sum(A_{C,C}) + 2sum(A_{C,\bar{C}}) + sum(diag(A_{\bar{C},\bar{C}})) \\
		&=& sum(A_{C,C}) + sum(2 \sum_{i \in C} A_{i,\bar{C}}+diag(A_{\bar{C},\bar{C}})) \\
		&=& sum(A_{C,C}) + sum( v_{\bar{C}}) \\
		& = & val_C
	\end{eqnarray*}
	where $\vec{v} = 2\sum_{i\in C} A[i,:]+  diag(A)$ and $v_{\bar{C}}$ is the sub-vector of $v$ with indices in $\bar{C}$. Given $C$, $sum(A_{C,C})$ is fixed, therefore $\bar{C}$ must be the set of indices of the largest $k-|C|$ elements from $\{ v_{m+1},...,v_{n} \}$. Algorithm 1 then loop over all the possible $C \subseteq [m]$ ($2^m$ many ) and identifies the optimal one, i.e., the one achieving the maximum $val_C$. 
\end{proof}

\begin{algorithm}
	\begin{algorithmic}[1]
	\REQUIRE matrix $A$ of form \eqref{eqn:A}.		
	\ENSURE  a pure strategy $s$.
	\vspace{1mm}		
	\FOR{all $C\subseteq [m]$ constrained by $|C|\leq k$}
	\STATE	$\vec{v} = 2\sum_{i\in C} A[i,:]+ diag(A)$; 
	\STATE	Choose the largest $k-|C|$ values from the set $\{v_{m+1},...,v_n\}$, and denote the set of their indices as $\bar{C}$; 
	\STATE	Set $val_C = sum(A_{C,C})+sum(v_{\bar{C}})$; 
	\ENDFOR	
	\RETURN the pure strategy $s = C \cup \bar{C} $ with maximum $val_C$.
	\end{algorithmic}
	\caption{Defender Oracle}
	\label{alg:DefenderOracle}
\end{algorithm}


\begin{theorem}\label{thm:smallSupport}
	({\bf Polynomial Solvability}) There is an efficient $poly(n,k)$ time algorithm which computes the optimal defender strategy in the PRIL and ADIL model, if the size of the leakage support $m$ is a constant.
\end{theorem}

\subsection{An Approximation Algorithm}\label{sec:approx}
We now consider approximation algorithms. Recall that information leakage is due to the correlation between targets, thus one natural way to minimize leakage is to allocate each resource {\it independently} with certain distributions. Naturally, the normalized marginal $\vec{x}^{*} / k$ becomes a choice, where $\vec{x}^{*}$ is the  solution to LP~\eqref{lp:approxOPT}. To avoid the waste of using multiple resources to protect the same target, we sample without replacement. Formally, the {\it independent sampling without replacement} algorithm proceeds as follows: 1. compute the optimal solution $\vec{x}^*$ of LP~\eqref{lp:approxOPT}; 2. independently sample $k$ elements from $[n]$ { \it without replacement} using distribution $\vec{x}^*/k$.

Zero-sum games exhibit negative utilities, therefore an approximation ratio in terms of utility is not meaningful. To analyze the performance of this algorithm 
we shift all the payoffs by a constant, $-\min_{i} c_i$, and get an equivalent constant-sum game with all non-negative payoffs. Theorem~\ref{thm:uniformApprox} shows that this algorithm is ``almost'' a $(1-\frac{1}{e})- approximation$ to the optimal solution in the PRIL model, assuming information leaks out from any target $i$ with equal probability $p_i = \frac{1-p_0}{n}$. We note that proving a general approximation ratio for any $\vec{p}\in \Delta_{n+1}$ turns out to be very challenging, intuitively because the optimal strategy adjusts according to different $\vec{p}$ while the sampling algorithm does not depend on $\vec{p}$. However, experiments empirically show that the ratio does not vary much for different $\vec{p}$ on average (see Section~\ref{sec:exper}). 

\begin{theorem}\label{thm:uniformApprox}
	Assume each target leaks information with equal probability $p_i = \frac{1-p_0}{n}$. Let $\bar{c}_i \geq 0$ be the shifted cost and $U_{indepSample}$ be the defender utility achieved by \emph{ independent sampling without replacement}. Then we have:
	$$
	\begin{array}{c}
	U_{indepSample} \geq (\frac{k-2}{k-1} -\frac{1}{e}) \left[ Opt(LP~\ref{lp:naiveLP}) - (1-p_0)\frac{\sum_{i=1}^n \bar{c}_i}{n} \right].
	\end{array}
	$$
	where $(1-p_0)\frac{\sum_{i=1}^n \bar{c}_i}{n}$ is an additive loss to $Opt(LP~\ref{lp:naiveLP})$, which is usually small in security games.
\end{theorem}

\subsection*{Proof of Theorem~\ref{thm:uniformApprox}}
Let $Y=Y(\vec{x}) \in \RR^{n\times n}$ be a function of any $\vec{x}\in \RR^{n}$, where $y_{ij}$ is the probability that target $i,j$ are both protected using independent sampling {\it without} replacement. We first prove Lemma~\ref{lem:probbound}, which provides a lower bound regarding how good the pair-wise marginals in $Y$ approximate the given marginals $\vec{x}$. The difficulty of proving Lemma~\ref{lem:probbound} lies at that  $Y$ does not have a close form in terms of $\vec{x}$ if we sample without replacement. Our proof  is based on a coupling argument by relating the algorithm to independent sampling {\it with} replacement. 
\begin{lemma} \label{lem:probbound}
	Given $\vec{x}$, $Y=Y(\vec{x})$ satisfies the following (in)equalities:
	\begin{eqnarray}
	&& \sum_{i\in[n]} y_{ii} = k; \label{eqn:aproxProbmargin} \\
	&& y_{ii} \geq (1-\frac{1}{e}) x_i, \, \forall i \in [n]; \label{eqn:aproxProbentry}\\
	&& \frac{y_{ij}}{y_{ii}}  \geq (\frac{k-2}{k-1}-\frac{1}{e})x_j, \, \forall i, \not = j. \label{eqn:aproxProbbinary}
	\end{eqnarray}
\end{lemma}

\begin{proof}
	The first equation is easy to see, since each sampled pure strategy has $k$ different targets due to sampling without replacement. To prove the other two inequalities, we instead consider independent sampling {\it with} replacement. Similarly, define function $Z=Z(\vec{x})\in R^{n\times n}$ to a function of $\vec{x}$, where $z_{ij}$ is the probability that target $i,j$ are protected together when sampling with replacement. Contrary to $Y$, $Z$ has succinct close forms, therefore we can lower bound entries in $Z$. We first consider $z_{ii}$.
	\begin{eqnarray*}
	z_{ii} &=& 1-(1-x_{i}/k)^{k} \\
	&\geq & 1-e^{-x_i} \\ 
	& \geq & (1-\frac{1}{e})x_{i}.
	\end{eqnarray*}
	where we used the fact $(1-\epsilon)^{\frac{1}{\epsilon}} \leq e^{-1}$ for any $\epsilon \in (0,1)$. Now we lower bound $z_{ij}/z_{ii}$ as follows. 
	\begin{eqnarray} \label{eqn:cor_eqn}
	\frac{z_{ij}}{z_{ii}} &=& \frac{1-(1-\frac{x_{i}}{k})^{k}-(1-\frac{x_{j}}{k})^{k}+(1-\frac{x_{i}}{k}-\frac{x_{j}}{k})^{k}}{1-(1-x_{i}/k)^{k}} \\ \nonumber
	& = & 1-(1-\frac{x_{j}}{k})^{k} - \frac{(1-\frac{x_i}{k})^k(1-\frac{x_j}{k})^k-(1-\frac{x_{i}}{k}-\frac{x_{j}}{k})^{k}}{1-(1-x_{i}/k)^{k}}\\ \nonumber
	& = & 1-(1-\frac{x_{j}}{k})^{k} - \frac{(1-\frac{x_i}{k})^k}{1-(1-x_{i}/k)^{k}}[(1-\frac{x_j}{k})^k - (1-\frac{x_j}{k-x_i})^k]\\ \nonumber
	& \geq & (1-\frac{1}{e})x_j - \frac{e^{-x_i}}{1-e^{-x_i}}[(1-\frac{x_j}{k})^k - (1-\frac{x_j}{k-x_i})^k] 
	\end{eqnarray}
	where all the equations just follow the arithmetic, while the inequality uses the fact that $(1-\frac{x_j}{k})^k \leq e^{-x_j}$ and $-\frac{x}{1-x}$ is a decreasing function of $x \in (0,1)$.  
	We now upper-bound the term $(1-\frac{x_j}{k})^k - (1-\frac{x_j}{k-1})^k$ using the formula $a^k - b^k = (a-b) \sum_{i=0}^{k-1}a^i b^{k-1-i}$, as follows
	
\begin{eqnarray*}
	& & (1-\frac{x_j}{k})^k - (1-\frac{x_j}{k-x_i})^k \\
	&=& (1-\frac{x_j}{k} -1+\frac{x_j}{k-x_i})\sum_{t=0}^{k-1} (1-\frac{x_j}{k})^t(1-\frac{x_j}{k-x_i})^{k-1-t} \\
	& \leq & \frac{x_j x_i}{k(k-x_i)} \times k  \\
	& \leq & \frac{x_i x_j}{k-1}
\end{eqnarray*}
	Plugging in the above upper bound back to Inequality~\ref{eqn:cor_eqn}, we thus have
	\begin{eqnarray*}
		\frac{z_{ij}}{z_{ii}} & \geq & (1-\frac{1}{e})x_j - \frac{e^{-x_i}}{1-e^{-x_i}} \frac{x_i x_j}{k-1} \\
		& \geq & (1-\frac{1}{e})x_j - \frac{x_i}{e^{x_i}-1} \frac{x_j}{k-1} \\
		& \geq & (1-\frac{1}{e})x_j - \frac{x_j}{k-1} \\
		& = & (\frac{k-2}{k-1}-\frac{1}{e})x_j
	\end{eqnarray*}
	where the last inequality is due to the fact that $f(x)=\frac{x}{e^x -1}$ is a decreasing function for $x\in(0,1)$ and is upper bounded by $\lim_{x\to 0} \frac{x}{e^x-1}=1$.  
	
	Therefore, we have $\frac{z_{ij}}{z_{ii}}  \geq (\frac{k-2}{k-1}-\frac{1}{e})x_j$. We then conclude our proofs by claiming that $y_{ii} \geq z_{ii}$ and $y_{ij}/y_{ii} \geq z_{ij}/z_{ii}$.
	
	To prove our claim, we use a coupling argument. Consider the following two stochastic process (StoP):
	\begin{enumerate}
		\item $StoP^1$: at time $t$ independently sample a random value $i_t$ ($\in [n]$) with probability $x_{i_t}/k$ for any $t=1,2,...$ until precisely $k$ {\it different} elements from $[n]$ show up.
		
		\item  $StoP^2$: at time $t$ independently sample a random value $i_t$ ($\in [n]$) with probability $x_{i_t}/k$ for $t=1,2,...k$. 
	\end{enumerate}
	Let $C^1$ [$C^2$] denote all the possible random sequences generated by $StoP^1$ [$StoP^2$], and $C_{i}^1$ [$C_{i}^2$] denote the subset of $C^1$ [$C^2$], which consists of all the sequences including at least one $i$. For any $e\in C_{i}^2$, let $C_{e}$ be the subset of sequences in $C^1$, whose first $k$ element is precisely $e$. Notice that any sequence in $C^1$ has at least length of $k$ while any sequence in $C^2$ has precisely $k$ elements. Furthermore, $C_e \subseteq C_{i}^1$ and $C_e \cap C_{e'} = \emptyset$ for any $e,e' \in C_{i}^2$ and $e \not = e'$. 
	
	Now, think of each sequence as a probabilistic event generated by the stochastic process. Notice that $P(e;StoP^2)=P(C_e;StoP^1)$ due to the independence of the sampling procedure, therefore, we have
	\begin{eqnarray*}
		P(C_{i}^2 ; StoP^2) &=& \sum_{e\in C_{i}^2} P(e;StoP^2) \\
		&=& \sum_{e\in C_{i}^2} P(C_e;StoP^1) \\
		&\leq & P(C_{i}^1 ; StoP^1)
	\end{eqnarray*}
	However, $P(C_{i}^1|StoP^1)=y_{ii}$ and $P(C_{i}^2 | StoP^2) = z_{ii}$. This proves $y_{ii} \geq z_{ii}$.
	
	Notice that  $y_{ij}/y_{ii} \geq z_{ij}/z_{ii}$ is equivalent to $P(e \in C_j^2 | e \in C_i^2; StoP^2) \geq P(e \in C_j^1 | e \in C_i^1; StoP^1)$. To prove this inequality, we claim that it is without loss of generality to assume the first sample is $i$ in both processes. This is because, if the first $i$ shows up at the $t$'th sample, moving $i$ to the first position would not change the probability of the sequence due to independence between each sampling step. Conditioned on $i$ is sampled first, a similar argument as above shows that the probability of Stochastic process $StoP^1$ generating $j$ is at least the probability of stochastic process $StoP^2$ generating $j$. 
\end{proof}

Let $\vec{x}^*$ be the optimal solution to LP~\eqref{lp:approxOPT} and $U^*$ be the corresponding objective value -- the defender optimal utility with no leakage. To prove Theorem 2, we start from comparing $OPT(LP~\ref{lp:naiveLP})$ with $U^*$. From the objective of LP \eqref{lp:naiveLP}, we know that $u \leq U^*$, $u_i \leq U^*$ since $U^*$ is the best possibly utility using $k$ resources, and $v_i \leq \bar{c}_i$ since if target $i$ is uncovered, the defender gets utility at most $\bar{c}_i$. Therefore, we have
\begin{eqnarray} \label{eqn:ind_upper} \nonumber
	OPT(LP~\ref{lp:naiveLP})&\leq& p_0 U^* + \frac{1-p_0}{n}\sum_{i=1}^n (x_{ii}^* U^* + (1-x_{ii}^*)\bar{c}_i) \\ 
	&\leq&  (p_0  + k\frac{1-p_0}{n}) U^* + \frac{1-p_0}{n} \sum_{i=1}^n \bar{c}_i
\end{eqnarray}
where we used the equation $\sum_{i \in [n]} x_{ii}^* = k$. We now examine $U_{indepSample}$. A simple argument yields that, if $\bar{c}_i \geq 0$ for all $i$ and each target $i$ is covered by probability at least $\alpha x_i^*$ for any $i$, then the defender utility is at least $\alpha U^*$. Therefore, by Lemma~\ref{lem:probbound} we have
\begin{eqnarray} \label{eqn:ind_lower} \nonumber
	 U_{indepSample} &\geq& p_0 (1-\frac{1}{e}) U^* + \frac{1-p_0}{n}\sum_{i=1}^n y_{ii} (\frac{k-2}{k-1}-\frac{1}{e})U^* \\ 
	&\geq&  (\frac{k-2}{k-1}-\frac{1}{e})(p_0  + k\frac{1-p_0}{n}) U^*
\end{eqnarray}
where we used the fact that $\sum_{i \in [n]} y_{ii}^* = k$. 
Comparing Inequalities \eqref{eqn:ind_upper} and \eqref{eqn:ind_lower}, we have $U_{indepSample} \geq (\frac{k-2}{k-1} -\frac{1}{e})[OPT(LP2) - \frac{1-p_0}{n} \sum_{i=1}^n \bar{c}_i]$. This concludes our proof of Theorem~\ref{thm:uniformApprox}. 


\section{Sampling Algorithms}

From Carath\'{e}odory's theorem we know that, given any marginal coverage $\vec{x}$, there are many different mixed strategies achieving the same marginal $\vec{x}$ (e.g., see examples in Section~\ref{sec:model}). Another way to handle information leakage is to generate the optimal marginal coverage $\vec{x}^*$, computed by LP~\eqref{lp:approxOPT}, with low correlation between targets. Such a ``good" mixed strategy, e.g., the mixed strategy with maximum entropy, is usually supported on a pure strategy set of exponential size. 
In this section, we propose two sampling algorithms, which efficiently generate a mixed strategy with exponentially large support and are guaranteed to achieve any given marginal $\vec{x}$.


\subsection{Max-Entropy Sampling}
Perhaps the most natural choice to achieve low correlation is the distribution with maximum entropy restricted to achieving the marginal $\vec{x}$, 
which can be formulated as the solution of Convex Program (CP)~\eqref{lp:maxentropy}. However, naive approaches for CP~\eqref{lp:maxentropy} require exponential running time since there are  $O(2^n)$ variables.
Interestingly, it turns out that this be resolved.
\begin{lp} \label{lp:maxentropy}
	\maxi{\sum_{s \in S}-\theta_s \ln(\theta_s)} 
	\st
	\qcon{ \sum_{s:\, i\in s} \theta_s = x_i}{i \in [n]}
	\con{\sum_{s \in S} \theta_s =1}   
	\qcon{\theta_{s} \geq 0}{ s \in S}   
\end{lp}
where variable $\theta_s$ is the probability of taking pure strategy $s$.
\begin{theorem}\label{thm:maxentropy}
	There is an efficient algorithm which runs in $poly(n,k)$ time and outputs a pure strategy $s$ with probability $\theta_s^*$  for any pure strategy $s \in S$, where $\vec{\theta}^*$ is the optimal solution to Convex Program~\eqref{lp:maxentropy} (within machine precision\footnote{Computers cannot solve general convex programs exactly due to possible irrational solutions. Therefore, our algorithm is optimal within machine precision, and we simply call it "solved".}).
\end{theorem}
The proof of Theorem~\ref{thm:maxentropy} relies on Lemmas \ref{lem:entropycompute} and \ref{lem:entropysample}. Lemma~\ref{lem:entropycompute} presents a compact representation of $\vec{\theta}^{*}$ based on the KKT conditions of CP~\eqref{lp:maxentropy} and its dual -- the {\it unconstrained} Convex Program~\eqref{lp:maxentropydual}:  
\begin{lp}\label{lp:maxentropydual}
	\mini{f(\vec{\beta})=\sum_{i=1}^{n} \beta_i x_i + \ln (\sum_{s \in S} e^{-\beta_s}),}
\end{lp}
where variables $\vec{\beta} \in \mathbb{R}^n$ and $e^{-\beta_s}=\Pi_{i\in s} e^{-\beta_i}$. We notice that the dual program ~\eqref{lp:maxentropydual} as well as the characterization of $\theta_s^*$ in Lemma~\ref{lem:entropycompute} are not new (e.g., see \cite{EntropyMohit}), and we state it for completeness. Our contribution lies at proving that CP~\eqref{lp:maxentropydual} can be computed efficiently in $poly(n,k)$ time in our security game setting despite the summation $\sum_{s \in S} e^{-\beta_s}$ of $O(2^k)$ terms. 

\begin{lemma}\label{lem:entropycompute}
	Let $\vec{\beta}^* \in \mathbb{R}^n$ be the optimal solution to CP~\eqref{lp:maxentropydual} and set $\alpha_i =e^{-\beta_i^*}$ for any $i\in[n]$, then the optimal solution of CP~\eqref{lp:maxentropy} satisfies
	\begin{equation}
	\label{eq:entropyterm}
	\theta_s^{*}= \frac{\alpha_s}{\sum_{s\in S} \alpha_s},
	\end{equation}
	where $\alpha_s = \Pi_{i\in s} \alpha_i$ for any pure strategy $s \in S$. 
	
	Furthermore, $\vec{\beta}^*$ can be computed in $poly(n,k)$ time.
\end{lemma}

\begin{proof}
	As proved in \cite{EntropyMohit}, the $\alpha^i$ above is precisely $e^{-\beta_i^{*}}$ where $\vec{\beta}^*$ is the optimal solution to CP~\eqref{lp:maxentropy}. We show that $\vec{\beta}^*$ can be computed in $poly(n,k)$ time. Notice that CP~\eqref{lp:maxentropydual} has $n$ variables but an expression of exponentially many terms, specifically, $\sum_{s\in S} e^{-\beta_s}$. The essential difficulty of computing $f(\vec{\beta})$ lies at computing the sum $\sum_{s\in S} e^{-\beta_s}$, since the other parts can be explicitly calculated in polynomial time.  Fortunately the sum $\sum_{s\in S} e^{-\beta_s}$ exhibits some combinatorial structure,and combinatorial algorithms could be employed for computation. In particular, we show that a dynamic program computes the sum $\sum_{s\in S} e^{-\beta_s}$ in $poly(n,k)$ time. The algorithm for  computing $\nabla f(\vec{\beta})$ can be designed in a similar fashion, and hence left to the reader. Since a convex program can be solved efficiently in machine precision given the access to its function value and derivatives, we then conclude our proof by describing the following dynamic program to compute $\sum_{s\in S} e^{-\beta_s}$, given any $\vec{\beta}$. 
	
	Notice that the set of all pure strategies consists of all the subsets of $[n]$ of cardinality $k$. Let $\alpha_i = e^{-\beta_i}$ and  $\alpha_s = \Pi_{i\in s} \alpha_i$. We then build the following DP table $T(i,j)=\sum_{s:s\subseteq [j],|s|=i} \alpha_s$, which sums over all the subsets of $[j]$ of cardinality $i$. Our goal is to compute $T(k,n)=\sum_{s\in S} e^{-\beta_s}$. We first initialize $T(0,j)=1$ and $T(j,j)=\Pi_{i=1}^j \alpha_i$  for any $j$. Then using the following update rule, we can build the DP table and compute $T(k,n)$ in $poly(k,n)$ time.
	\begin{eqnarray*}
		T(i,j) &=& T(i,j-1)+\alpha_j T(i-1,j-1).
	\end{eqnarray*}
\end{proof}

Our next lemma considers how to efficiently sample a pure strategy $s$ from an exponentially large support with probability $\theta_s^*$ represented by Equation~\eqref{eq:entropyterm}. The algorithm, as detailed in Algorithm~\ref{alg:EntropSampling}, simply goes through each target and adds it to the pure strategy with a specifically designed probability until exactly $k$ targets are added.

\begin{algorithm}
	\begin{algorithmic}[1]
		\REQUIRE :  $\vec{\alpha}\in[0,\infty)^n$, $k$.
		
		\ENSURE : a pure strategy $s$ with $|s|=k$.
		\vspace{1mm}
		
		\STATE Initialize: $s = \emptyset$; the DP table $T(0,j)=1$ and $T(j,j)=\Pi_{i=1}^j \alpha_i$  for any $j\in [n]$.
		
		\STATE Compute $T(i,j)=\sum_{s:s\subseteq [j],|s|=i} \alpha_s$ for any $i,j$ satisfying $i\leq k,j\leq n$ and $1 \leq i\leq j$, using the following update rule 
		\begin{eqnarray*}
			T(i,j) &=& T(i,j-1)+\alpha_j T(i-1,j-1).
		\end{eqnarray*}
		\STATE Set $i=k$, $j=n$;
		
		\WHILE{$i>0$}
		\STATE Sampling: independently add $j$ to $s$ with probability \begin{equation*}
		p=\frac{\alpha_j T(i-1,j-1)}{T(i,j)};
		\end{equation*}
		\IF{$j$ added to $s$} 
		\STATE $i = i - 1$;
		\ENDIF
		\STATE $j=j-1$;         
		\ENDWHILE
		\RETURN s.
	\end{algorithmic}
	\caption{Max-Entropy Sampling 	\label{alg:EntropSampling}}
\end{algorithm}

\begin{lemma}\label{lem:entropysample}
	Given any input $\vec{\alpha} \in [ 0,\infty)^n$, Algorithm~\ref{alg:EntropSampling} runs in $poly(k,n)$ time and correctly samples a pure strategy $s$ with probability $\theta_s = \frac{\alpha_s}{\sum_{s\in S} \alpha_s}$, where $\alpha_s = \Pi_{i\in s} \alpha_i$.
\end{lemma}
\begin{proof}
It is easy to see that Table $T(i,j)$ can be computed in $poly(n,k)$.  We first show that the ``while" loop in Algorithm~\ref{alg:EntropSampling} terminates within at most $n$ steps. In fact, $j$ decreases by $1$ each step and furthermore $j\geq i\geq 0$ always holds.  This is because when $j$ decreases until $j=i$, $j$ will  be sampled with probability $\frac{\alpha_j T(i-1,j-1)}{T(i,j)}=\frac{\alpha_i T(i-1,i-1)}{T(i,i)} = 1$; then both $j$ and $i$ will decrease by $1$ (Step $6-9$). This continues until $i=0$. Furthermore, the algorithm terminates with $|s| = k$ because the cardinality of $s$ always satisfies $|s|=k-i$ by Step $6-8$ until the termination at $i=0$. Therefore, Algorithm~\ref{alg:EntropSampling} runs in $poly(n,k)$ time.
	
	Now we show that Algorithm~2 outputs $s$ with probability $\theta_s$. Let the output $s=\{i_1,...,i_k\}$ be sorted in decreasing order, i.e., $i_1 > i_2 >...>i_k$. Notice that
	\begin{equation*}
	T(i,j)=\alpha_j T(i-1,j-1)+T(i,j-1).
	\end{equation*}
	Therefore, in the {\it Sampling} step (Step $5$) of Algorithm~\ref{alg:EntropSampling}, $j$ is not included to $s$ with probability $T(i,j-1)/T(i,j)$. Therefore, to sample $s=\{i_1,...,i_k\}$, it must be the case that $n,n-1,...,i_1+1$  are not included, while $i_1$ is included; $i_1 -1,...,i_2+1$ are not included, while $i_2$ is included; and so on so forth. In addition, the sampling in each of these steps is independent and the probability of each  step is known. Therefore, by multiplying these probabilities together, we have 
	\begin{eqnarray*}
		P(s) &=& \frac{T(k,n-1)}{T(k,n)}\times \frac{T(k,n-2)}{T(k,n-1)} ... 
		\times \frac{\alpha_{i_1}T(k-1,i_1-1)}{T(k,i_1)} \\
		& & \times \frac{T(k-1,i_1-2)}{T(k-1,i_1-1)} ...\frac{\alpha_{i_k} T(0,i_k-1)}{T(1,i_k)}  \\
		&=& \frac{  \Pi_{t\leq k} \alpha_{i_t}  }{T(k,n)}  \\
		&=& \theta_s
	\end{eqnarray*}
This gives precisely the probability we want.
\end{proof}

\emph{Remark:} we notice that {\it approximately uniform} sampling from combinatorial structures has been studied in theoretical computer science \cite{Sample}. Algorithm~\ref{alg:EntropSampling} uses a variant of the algorithm in \cite{Sample}, and extends their results  to the {\it weighted} (by $\theta_s^*$) and {\it exact} case.

\subsection{Uniform Comb Sampling}\label{sec:comb}
\cite{Tsai10a} presented the Comb Sampling algorithm, which randomly samples a pure strategy and achieves a given marginal in expectation. The algorithm can be elegantly described as follows (also see Figure~\ref{fig:comb}): thinking of $k$ resources as $k$ buckets with height 1 each, we then put each target, the height of which equals precisely its marginal  probability, one by one into the buckets. If one bucket gets full when filling in a certain target, we move the ``rest" of that target to a new empty bucket. Continue this until all the targets are filled in, at which time we know that $k$ buckets are also full. The algorithm then takes a horizontal line with a uniformly randomly chosen height from the interval $[0,1]$, and the $k$ targets intersecting the horizontal line constitute the sampled pure strategy. As easily observed, Comb Sampling achieves the marginal coverage in expectation \cite{Tsai10a}.
\begin{figure}
	\begin{center}
		\includegraphics[bb=00bp 50bp 670bp 470bp,clip,scale=0.4]{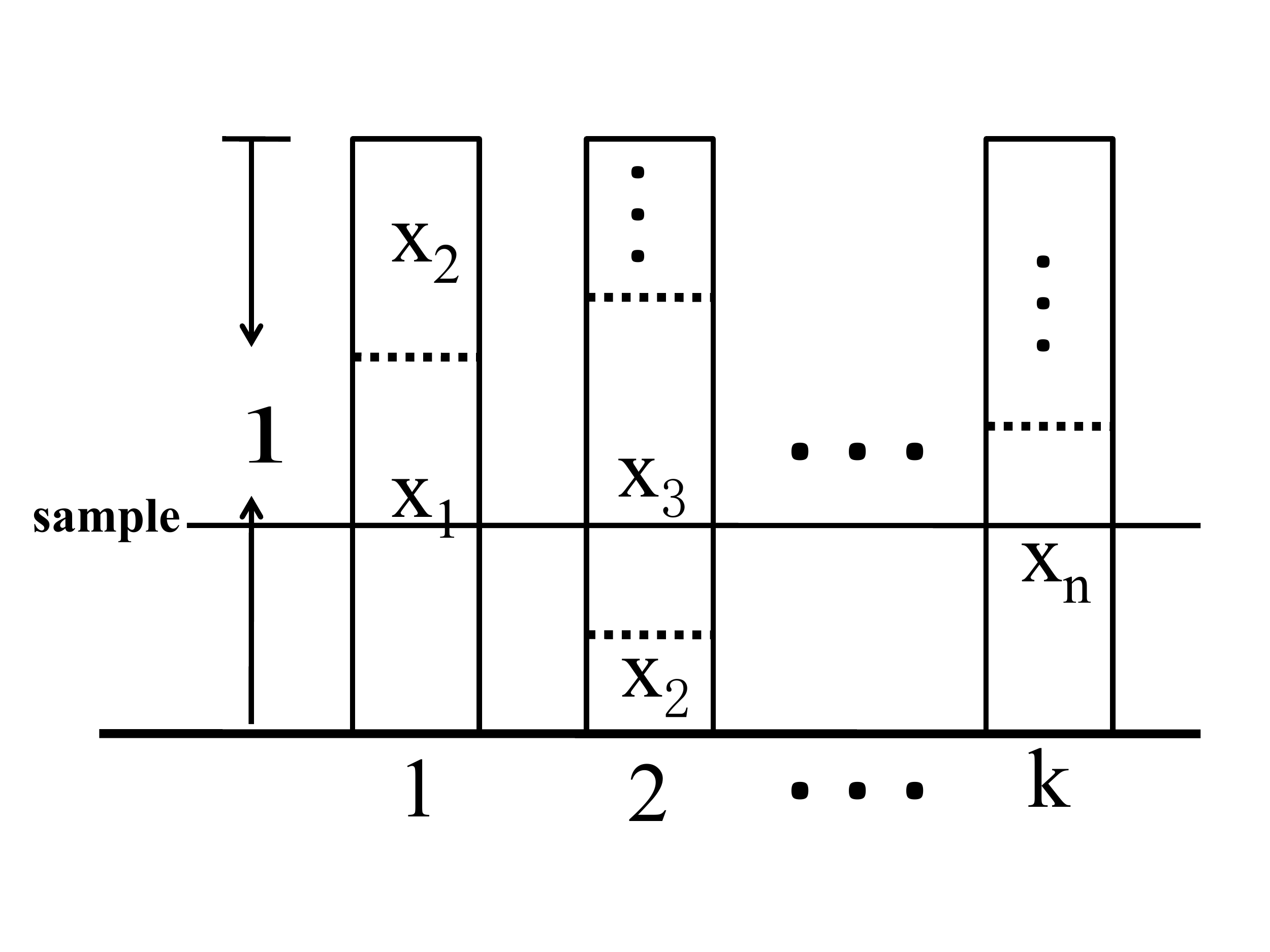}
	\end{center}
	\caption{\label{fig:comb} Comb Sampling}
\end{figure} 

However, is Comb Sampling robust against information leakage? We first observe that Comb Sampling generates a mixed strategy with support size at most $n+1$, which precisely matches the upper bound of Carath\'{e}odory's theorem.
\begin{proposition}\label{prop:CSsupport}
	Comb Sampling generates a mixed strategy which mixes over at most $n+1$ pure strategies.
\end{proposition} 
Proposition~\ref{prop:CSsupport} suggests that the mixed strategy sampled by Comb Sampling might be very easy to explore. Therefore we propose a variant of the Comb Sampling algorithm. Our key observation is that Comb Sampling achieves the marginal coverage regardless of the order of the targets. That is, the marginal is still obtained if we randomly shuffle the order of the targets {\it each time} before sampling, and then fill in them one by one. Therefore, we propose the following Uniform Comb Sampling (UniCS) algorithm:
\begin{enumerate}
	\item Order the $n$ targets uniformly at random; 
	\item fill the targets into the buckets based on the random order, and then apply Comb Sampling.
\end{enumerate}
Since the order is chosen randomly  each time, the mixed strategy implemented by UniCS mixes over exponentially many pure strategies, and achieves the marginal.
\begin{proposition}
	Uniform Comb Sampling (UniCS) achieves the marginal coverage probability.
\end{proposition}

\section{Experiments}\label{sec:exper}
Traditional algorithms for computing Strong Stackelberg Equilibrium (SSE) only optimize the coverage probability at each target, without considering their correlations. In this section, we experimentally study how traditional algorithms and our new algorithms perform in presence of {\it probabilistic} or {\it adversarial} information leakage. In particular, we compare the following five algorithms. 
\begin{itemize}
\item {\it Traditional}: optimal marginal + comb sampling, the traditional way to solve security games with no scheduling constraints \cite{Kiekintveld09,Tsai10a};

\item {\it OPT}: the optimal algorithm for PRIL or ADIL model (Section 3.1) using column generation with the defender oracle in Algorithm~\ref{alg:DefenderOracle};

\item {\it indepSample}: independent sampling without replacement (Section~\ref{sec:approx});

\item {\it MaxEntro}: max entropy sampling (Algorithm~\ref{alg:EntropSampling});

\item {\it UniCS}: uniform comb sampling (Section 4.2).
\end{itemize}
All algorithms are tested on the following two sets of data:

{\bf Los Angeles International Airport (LAX) Checkpoint Data} from \cite{Armor}. This problem was modeled as a Bayesian Stackelberg game with multiple adversary types in \cite{Armor}. To be consistent with our model, we instead only consider the game against one particular type of adversary -- the terrorist-type adversary, which is the main concern of the airport. The defender's rewards and costs are obtained from \cite{Armor} and the game is assumed to be zero-sum in our experiments.

{\bf Simulated Game Payoffs}. A systematic examination is conducted with simulated  payoffs. All  generated games have 20 targets and 10 resources. The reward  $r_i$ (cost $c_i$) of each target $i$ is chosen uniformly at random from the interval $[0,10]$ ($[-10,0]$).  
\begin{figure}
	\centering
	\begin{minipage}[t]{0.45\columnwidth}%
		\includegraphics[bb=60bp 185bp 670bp 610bp,clip,scale=0.45]{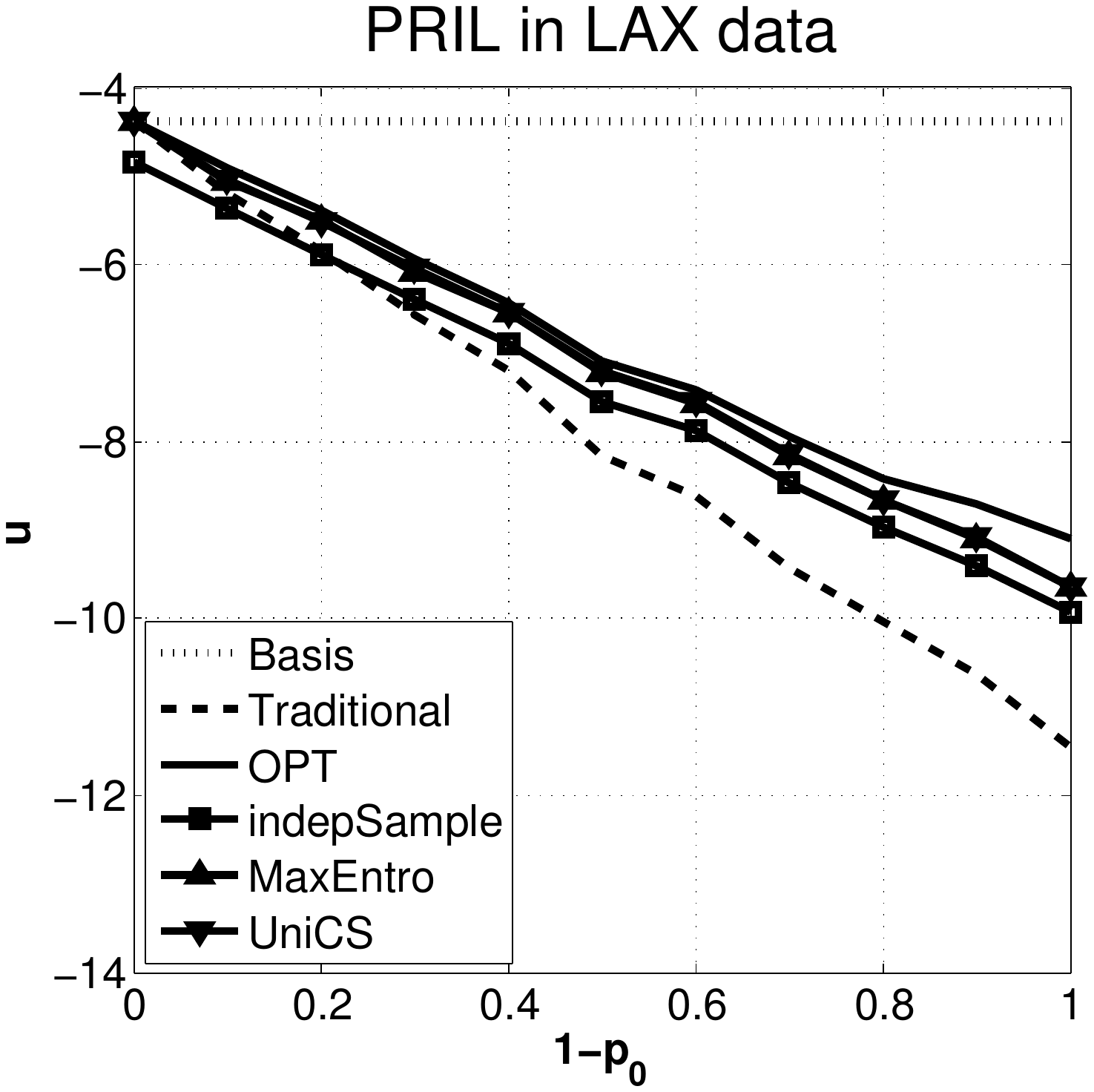}%
	\end{minipage}\qquad{}%
	\begin{minipage}[t]{0.45\columnwidth}%
		\includegraphics[bb=97bp 185bp 670bp 610bp,clip,scale=0.45]{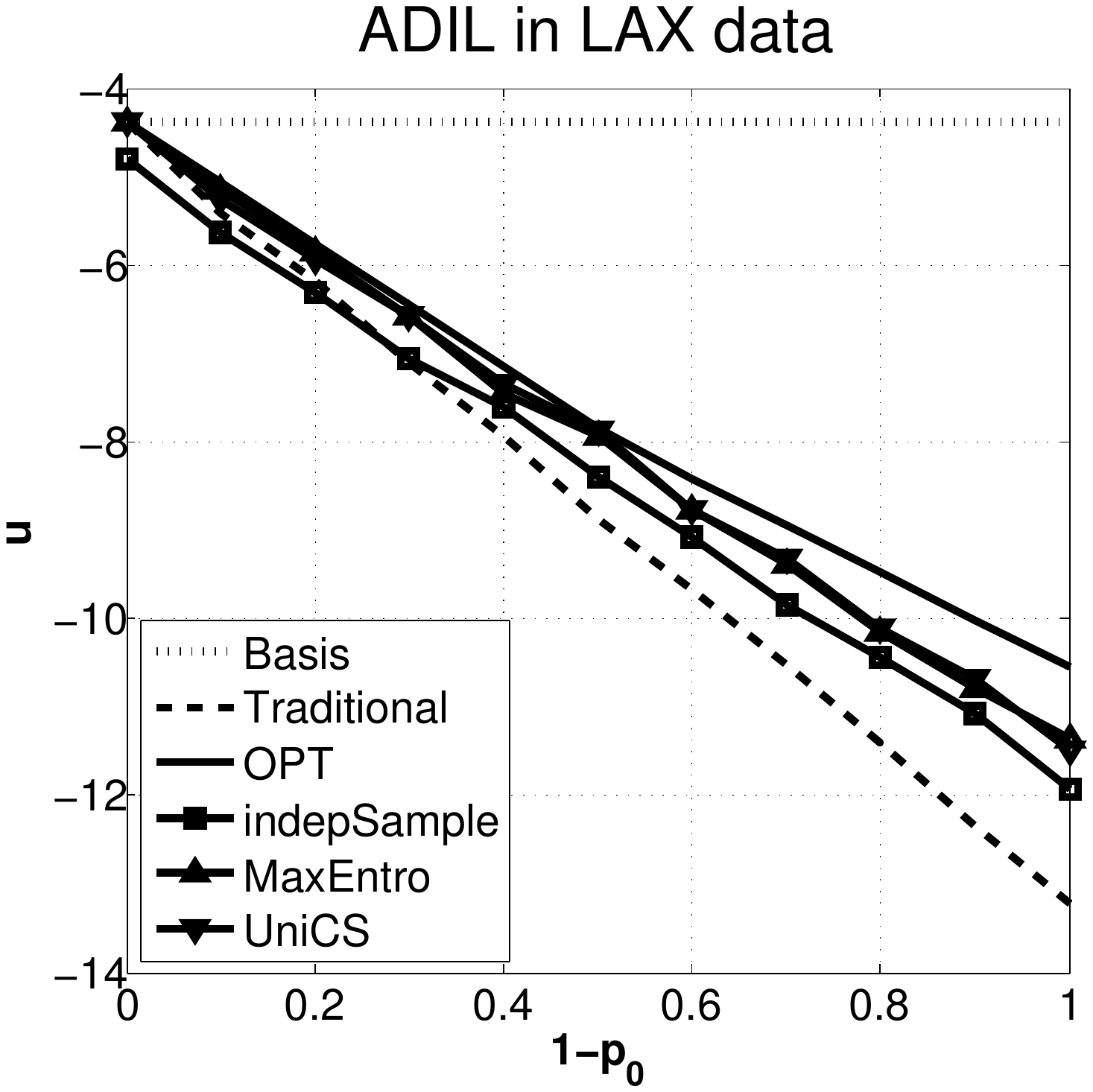}%
	\end{minipage}
	\caption{\label{fig:armor} Comparisons on real LAX airport data. }
\end{figure} 

\begin{figure}
	\begin{minipage}[t]{0.45\columnwidth}%
		\includegraphics[bb=60bp 185bp 670bp 610bp,clip,scale=0.45]{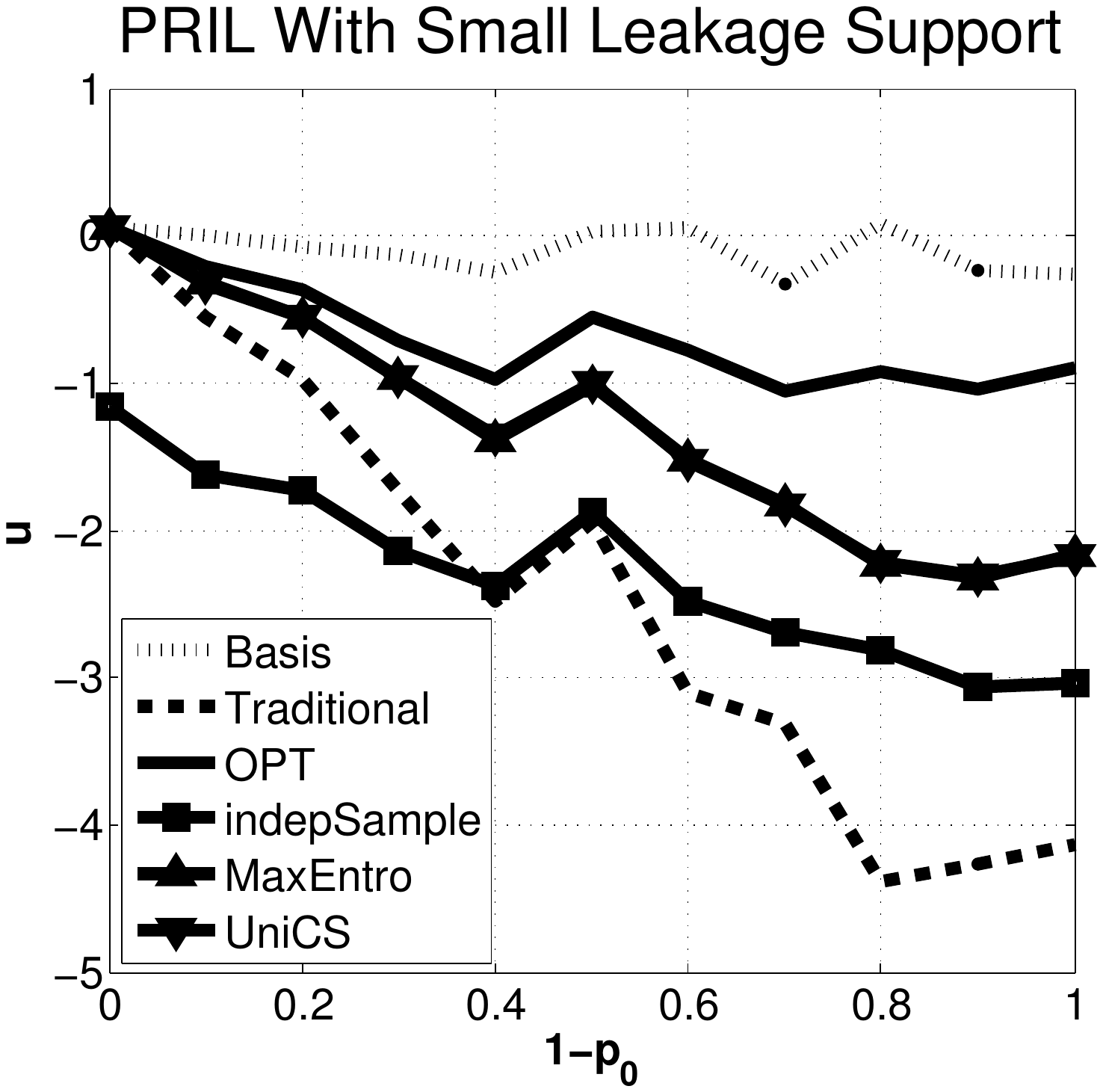}%
	\end{minipage}\qquad{}%
	\begin{minipage}[t]{0.45\columnwidth}%
		\includegraphics[bb=99bp 185bp 670bp 610bp,clip,scale=0.45]{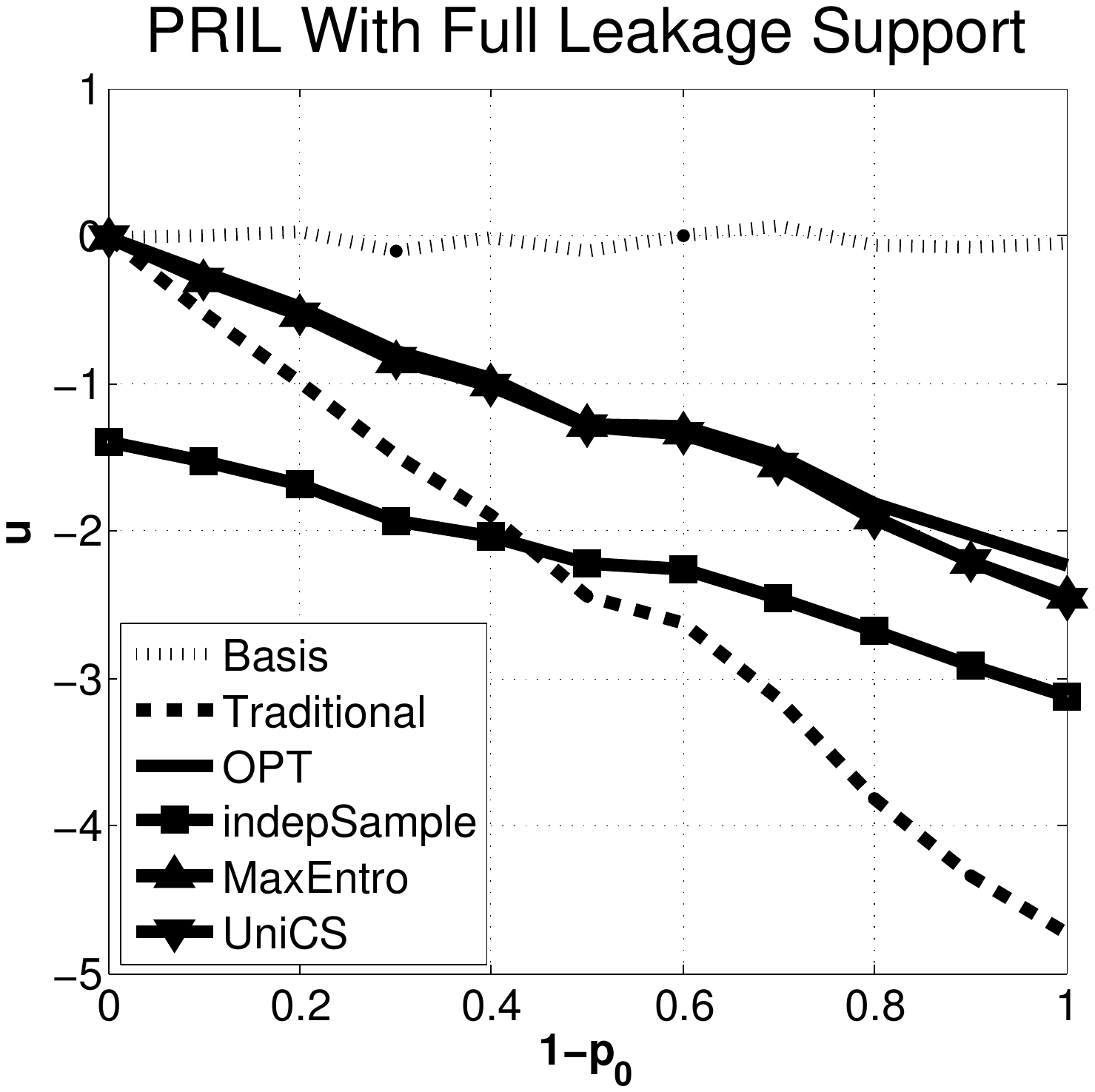}%
	\end{minipage}
	\begin{minipage}[t]{0.45\columnwidth}%
		\includegraphics[bb=60bp 185bp 670bp 610bp,clip,scale=0.45]{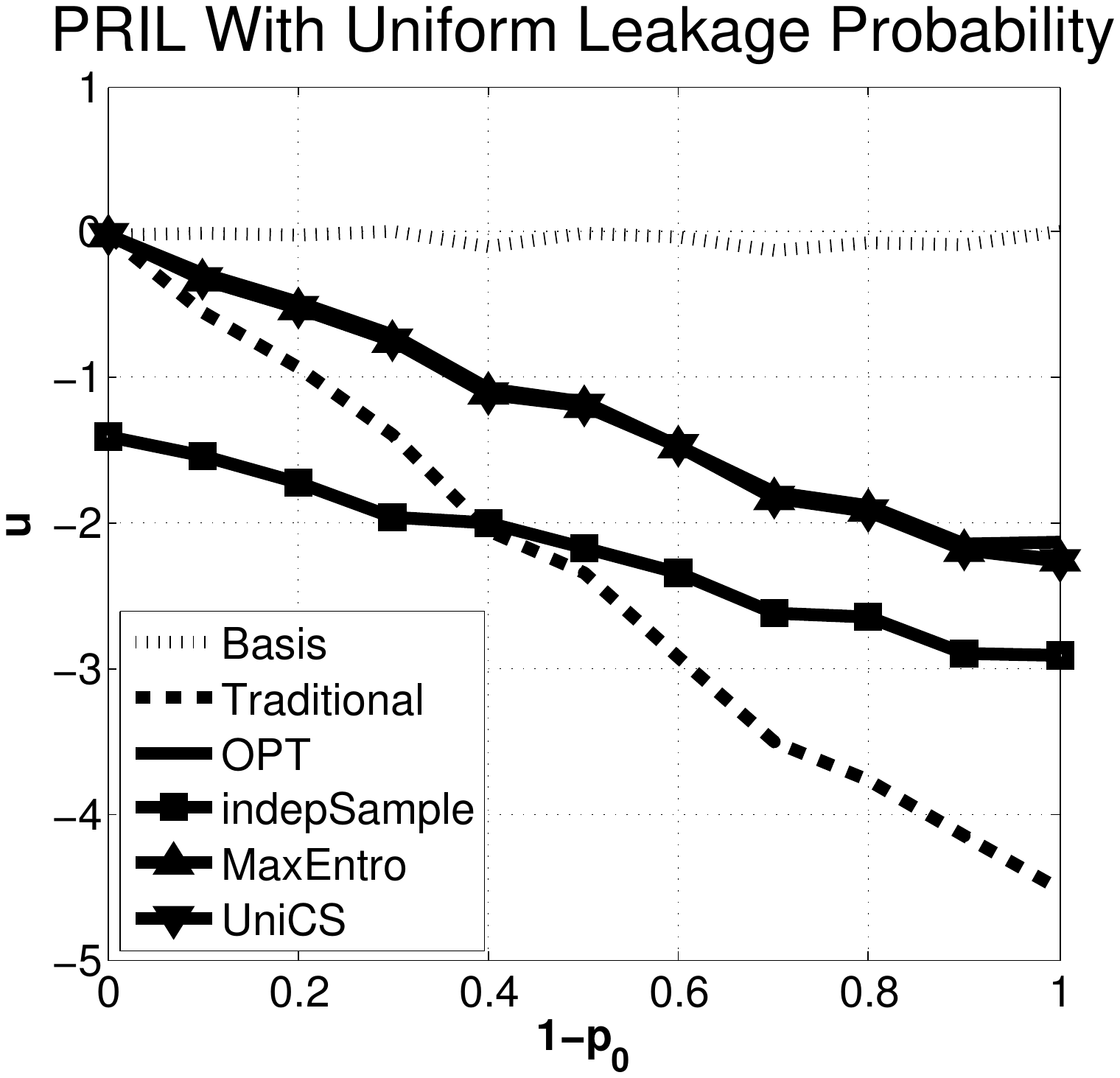}%
	\end{minipage}\qquad{}%
	\begin{minipage}[t]{0.45\columnwidth}%
		\includegraphics[bb=99bp 185bp 670bp 610bp,clip,scale=0.45]{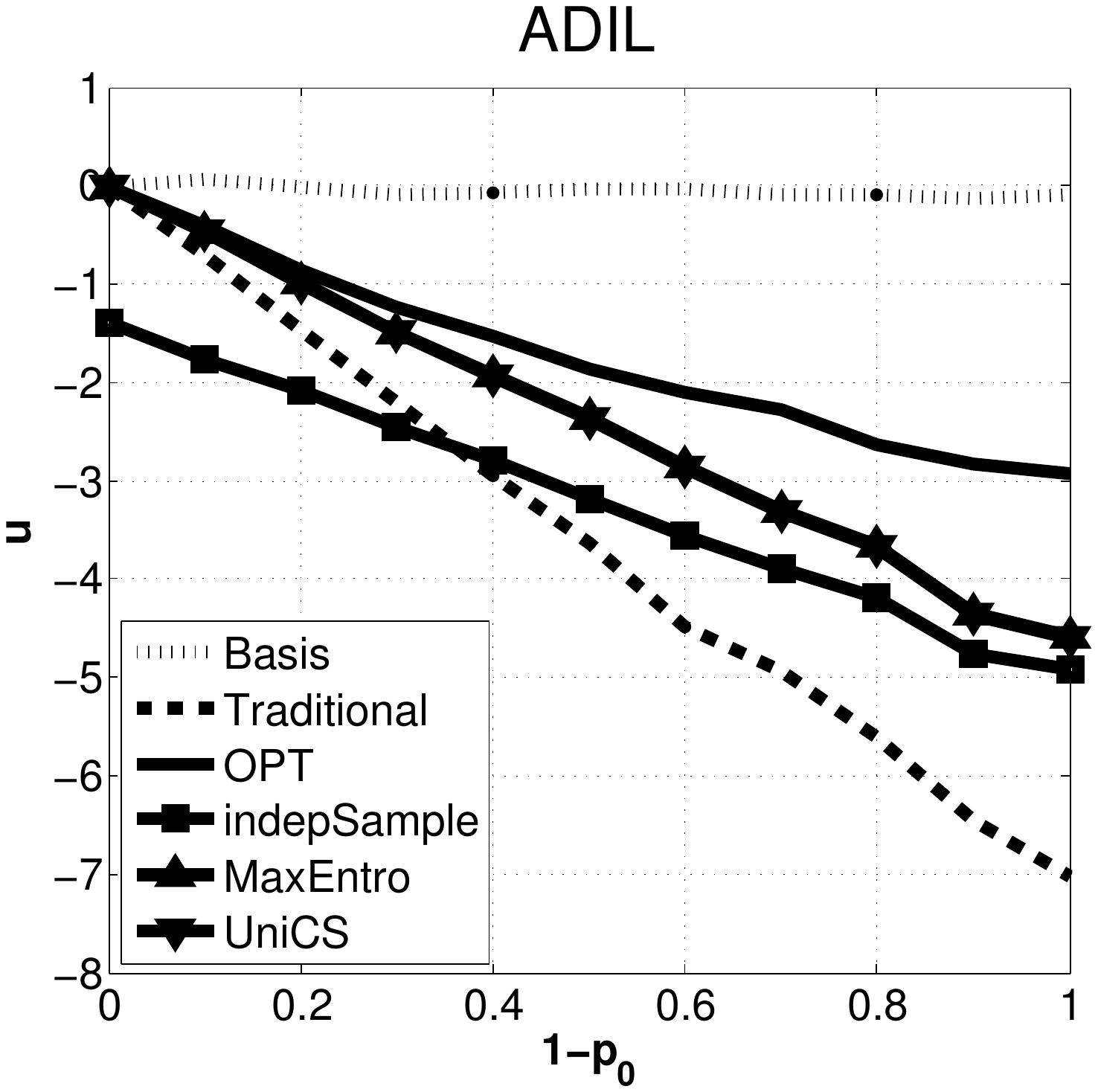}%
	\end{minipage}
	\caption{\label{fig:sim} Comparisons in Simulated Games. }
\end{figure}

In terms of running time, all the algorithms run efficiently as expected (terminate within seconds using MATLAB) except the optimal algorithm {\it OPT}, which takes about $3$ minutes per simulated game on average. Therefore we mainly compare defender utilities. All the comparisons are listed in Figure~\ref{fig:armor} (for LAX data) and Figure~\ref{fig:sim} (for simulated data). The line ``Basis" is the utility with no leakage and is listed as a basis for utility comparisons. Y-axis is the defender's utility -- the higher, the better. We examine the effect of the {\it total probability of leakage} (i.e., the x-axis $1-p_0$) on the defender's utility and consider $1-p_0 = 0, 0.1, ..., 1$. For probabilistic information leakage, we randomly generate the probabilities that each target leaks information with the constraint $\sum_{i=1}^n p_i = 1-p_0$. For the case of leakage from small support (for simulated payoffs only), we randomly choose a support of size $5$. All the utilities are \emph{averaged} over $50$ random games except the ADIL model for LAX data.
For the simulated payoffs, we also consider a special case of uniform leakage probability of each target (see Theorem~\ref{thm:uniformApprox}). The following observations follow from the figures.

{\bf Observation 1}. The gap between the line ``{\it Basis}" and ``{\it OPT}" shows that information leakage from even one target does cause dramatic utility decrease to the defender. Moreover, adversarial leakage causes more utility loss than probabilistic leakage; leakage from a restricted small support of targets causes less utility decrease than from full support.

{\bf Observation 2}. The gap between the line ``{\it OPT}" and ``{\it Traditional}" demonstrates the necessity of handling information leakage. In particular, the  relative loss $u(OPT)-u(Basis)$ is approximately half of the relative loss $u(Traditional)-u(Basis)$ in Figure~\ref{fig:sim} (and  $65 \%$ in  Figure~\ref{fig:armor}). Furthermore, if  leakage is from a small support (left-up panel in Figure~\ref{fig:sim}), {\it OPT} is close to {\it Basis}.  

{\bf Observation 3}. {\it MaxEntro} and {\it UniCS} have almost the same performance (overlapping in all these figures). Both algorithms are almost optimal when the leakage support is the full set $[n]$ (they almost overlap with {\it OPT} in the right-up and left-down panels in Figure~\ref{fig:sim}).

{\bf Observation 4}. An interesting observation is that {\it IndepSample} outperforms {\it Traditional} at $1-p_0=0.3$ or $0.4$ in all of these figures, which is around $\frac{1}{e} \approx 0.37$. Furthermore, the gap between {\it IndepSample}  and {\it OPT} does not change much at different $1-p_0$. 

{\bf Observation 5}. From a practical view, if the leakage is from a small support, {\it OPT} is preferred as it admits efficient algorithms (Section~\ref{subsec:smallSupport}); if the leakage is from a large support, {\it MaxEntropy} and {\it UniCS} are preferred as they can be computed efficiently and are close to optimality. From a theoretical perspective,  we note that the intriguing performance of {\it IndepSample}, {\it MaxEntropy} and {\it UniCS} raises questions for future work.
 
\section{Conclusions and Discussions}
In this paper, we considered partial information leakage in Stackelberg security games. We focused on the one-target leakage case, but do emphasize that  our models, hardness results and algorithms can be easily generalized. Our results raise several new research questions, e.g., is it possible to derive a theoretical approximation guarantee for {\it MaxEntro} and {\it UniCS}, and can we develop efficient algorithms to handle information leakage in other security game settings? 
More generally, it is an interesting problem to study analogous issues of information leakage in other settings beyond security, e.g., auctions or general games.    

\bibliography{refer}
\bibliographystyle{plain}

\section{Appendix A: Derivation of the Defender Oracle}
A defender oracle is a subroutine used to solve security games with large number of pure strategies by the \emph{Column Generation} technique. In this section, we first describe the technique of column generation and then derive the formulation for the defender oracle in our models.

Recall that LP~\eqref{lp:naiveLP} has a large number of variables because the number of pure strategies is exponential. However, by counting the number of activated constraints at optimality, we know that only polynomially many of these pure strategies will have non-zero probabilities at optimality since most pure strategies activate the corresponding constraint $\theta_s \geq 0$ and take probability $0$.  Column generation is based on this observation, i.e., the optimal mixed strategy has a small support. Basically, instead of solving LP~\eqref{lp:naiveLP} on the set $S$ of all pure strategies, it starts from a small subset of pure strategies, denoted as $\color{red}{A}$, and solve the following ``restricted" LP.

\begin{lp}\label{lp:CG_LP}
	\maxi{p_0 u+\sum_{i=1}^n p_i(u_i +v_i)}
	\st
	\qcon{ u\leq r_j x_{jj} +c_j (1-x_{jj}) }{j \in [n]}
	\qcon{u_{i}\leq r_j x_{ij}+c_j (x_{ii}-x_{ij}) }{ i,j \in [n]}
	\qcon{v_{i}\leq r_j (x_{jj}-x_{ij})+ c_j (1-x_{ii}- x_{jj}+x_{ij}) }{ i,j \in [n]}
	\qcon{ x_{ij}=\sum_{s \in  \textcolor{red}{A} : i,j \in s} \theta_s }{ i,j \in [n]}
	\con{ \sum_{s\in \color{red}{A}} \theta_{s} = 1}
	\qcon{\theta_s \geq 0}{  s \in \color{red}{A}}
\end{lp}
Notice that the only difference between LP~\eqref{lp:naiveLP} and LP~\eqref{lp:CG_LP} is that the set $S$ of all pure strategies is substituted by a small subset $\color{red}{A}$. In practice, $\color{red}{A}$ is usually initialized with a small number of pure strategies that are arbitrarily chosen. Column generation proceeds roughly as follows: 1. it solves LP~\eqref{lp:CG_LP}; 2. by checking the dual of LP~\eqref{lp:CG_LP}  the defender oracle judges whether the computed optimal solution to LP~\eqref{lp:CG_LP} is also optimal to LP~\eqref{lp:naiveLP} (setting all pure strategies in $S \setminus \color{red}{A}$ with probability $0$); if not, the oracle finds a new pure strategy to be added to the set $\color{red}{A}$ and updates $\color{red}{A}$. This procedure continues until the defender oracle judges that the computed optimal solution w.r.t. current $\color{red}{A}$ is also optimal to LP~\eqref{lp:naiveLP}. We now explain the underlying rationale of the column generation technique.

We first derive the dual of LP~\eqref{lp:CG_LP}. In fact, to emphasize the key aspects and avoid messy derivations, we  re-write LP~\eqref{lp:CG_LP} in the following abstract form:
\begin{lp}\label{lp:CG_compactLP}
	\maxi{d^T y}
	\st
	\con{ M x + N y \leq c }
	\qcon{ x_{ij} - \sum_{s \in  \textcolor{red}{A} : i,j \in s} \theta_s = 0 }{ i,j \in [n]}
	\con{ \sum_{s \in \color{red}{A}} \theta_{s} = 1}
	\qcon{\theta_s \geq 0}{  s \in \color{red}{A}}
\end{lp}
where variable $y$ represents the vector consisting of $u,v_i,u_i$ while variable $x$ is the vector representation of $x_{ij}$ (putting $i,j$ in some fixed order); $d$ is a vector summarizing the original objective coefficients; the constraints $M x + N y \leq c$ summarizes the first three  set of constraints in LP~\eqref{lp:CG_LP}. This abstract form not only simplifies our derivation of the dual, more importantly it emphasizes that the column generation technique works regardless of what the first three sets of constraints are as long as there are polynomially many of them. 

Let $M_{index(i,j)}$ be the \emph{column vector} of $M$ corresponding to variable $x_{ij}$ and $N_l$ be the \emph{column vector} of $N$ corresponding to the $l$'th component of $y$. We can now simply derive the dual of LP~\eqref{lp:CG_compactLP} as follows: 
\begin{lp}\label{lp:CG_compactdual}
	\mini{c^T \rho + \omega}
	\st
	\qcon{ \rho^T N_l \geq d_l }{\text{all } l}
	\qcon{ \rho^T M_{index(i,j)} + \beta_{ij \geq 0} }{ i,j \in [n]}
	\qcon{ - \sum_{i,j \in s} \beta_{ij} + \omega \geq 0 }{ s \in \color{red}{A}  }
	\con{\rho \geq 0}
\end{lp}
where $\rho$ are the dual variables w.r.t. the first set of constraints in LP~\eqref{lp:CG_compactLP} and $\beta_{ij}$, $\omega$ are the dual variables w.r.t. the second and third set of constraints.

First notice that the optimal solution to LP~\eqref{lp:CG_compactLP} (denoted as $OptSol_{\color{red}{A}}$) and the optimal solution to LP~\eqref{lp:CG_compactdual} (denoted as $OptSolDual_{\color{red}{A}}$) can both be computed efficiently when $A$ is small. A key observation here is that, if $OptSolDual_{\color{red}{A}}$, in particular, $\omega$ and $ \beta_{ij}$, happens to make the constraints $ - \sum_{i,j \in s} \beta_{ij} + \omega \geq 0, \, \forall s \in \color{red}{A}$ hold more generally as  $ - \sum_{i,j \in s} \beta_{ij} + \omega \geq 0, \, \forall s \in S$, then we claim that the $OptSol_{\color{red}{A}}$ is also an optimal solution  to LP~\eqref{lp:naiveLP} (by picking pure strategies in $S\setminus \color{red}{A}$ with probability $0$). This is because, if we substitute $\color{red}{A}$ by $S$ in both LP~\eqref{lp:CG_compactLP} and LP~\eqref{lp:CG_compactdual}, $OptSol_{\color{red}{A}}$ is still feasible to LP~\eqref{lp:CG_compactLP} because all the newly added strategies (in  $S\setminus \color{red}{A}$) have probability $0$; $OptSolDual_{\color{red}{A}}$ is still feasible to LP~\eqref{lp:CG_compactdual} because our $\omega, \beta_{ij}$ make constraints $ - \sum_{i,j \in s} \beta_{ij} + \omega \geq 0$ hold for all $ s \in S$ by assumption. Furthermore, \emph{complementary slackness} still holds since the added new variables in LP~\eqref{lp:CG_compactLP} all take value $0$. By linear program basics, we know that $OptSol_{\color{red}{A}}$ is still optimal if we substitute $\color{red}{A}$ in   LP~\eqref{lp:CG_compactLP} by $S$, which is precisely LP~\eqref{lp:naiveLP}.

As a result, our key task is to judge whether $ - \sum_{i,j \in s} \beta_{ij} + \omega \geq 0$ holds for all $s \in S$ for a given dual solution. This is equivalent to decide whether $\omega \geq \max_{s\in S} \left[ \sum_{i,j \in s} \beta_{ij} \right]$. The defender oracle is then defined as the following problem:
\begin{equation}\label{eqn:oracle}
\max_{s\in S} \left[ \sum_{i,j \in s} \beta_{ij} \right]= \max_{s \in S} s^T (\frac{M+M^T}{2})s
\end{equation}
 where $M$ is the matrix satisfying $M_{ij} = \beta_{ij}$. In other words, the defender oracle finds a pure strategy $s$ that maximizes the sum  $\sum_{i,j \in s} \beta_{ij}$.
 
 With this oracle, column generation proceeds, in more details, as follows: 1. compute LP~\eqref{lp:CG_compactLP} and LP~\eqref{lp:CG_compactdual}; 2. use the defender oracle to solve Problem~\eqref{eqn:oracle}: if the optimal value is less than or equal to the dual variable $\omega$, asserts optimality; otherwise, add $s^*$ -- the optimal solution to Problem~\eqref{eqn:oracle} -- to $\color{red}{A}$; 3. repeat until optimality is reached. Notice that the newly added $s^*$ does not belong to the original $\color{red}{A}$ because all $s \in \color{red}{A}$ satisfy $\sum_{i,j \in s} \beta_{ij} \leq \omega$. Column generation does not guarantee polynomial convergence, but usually converges very fast in practice. This is because the optimal mixed strategy usually has a small support. 
 
 When information leaks from a small subset of targets (Section~\ref{subsec:smallSupport}), the set of variables $x_{ij}$ become smaller since only correlation between leaking targets and other targets are considered. By modifying LP~\eqref{lp:CG_compactLP} and LP~\eqref{lp:CG_compactdual} a bit, we can get the defender oracle formulation as described in Section~\ref{subsec:smallSupport}.

\section{Appendix B: Relation between $\mathcal{P}(n,k)$ and the Correlation Polytope}

In this section, we show a connection between $\mathcal{P}(n,k)$ and the correlation polytope, defined as follows: 
\begin{definition}
	\cite{Corpoly} Given an integer $n$, the \emph{Correlation Polytope} $\mathcal{P}(n)$ is defined as follows
	\begin{equation*}
	\mathcal{P}(n)=Conv\left( \{v v^T: v\in \{0,1\}^n\} \right).
	\end{equation*}
	where $Conv(S)$ denotes the convex hull of set $S$. Notice that $vv^T \in \{0,1\}^{n \times n}$.
\end{definition}

The following lemma captures the relation between $\mathcal{P}(n)$ and
$\mathcal{P}(n,k)$.
\begin{lemma}
	\label{lem: polytope property}$X \in \mathcal{P}(n,k)$ if and only if the following three constraints hold: (a) $ X\in \mathcal{P}(n)$; (b) $tr(X)= \sum_{i=1}^{n} x_{ii} = k$; and (c)  $sum(X)=\sum_{i,j=1}^{n} x_{ij}=k^{2}$. In other words, $\mathcal{P}(n,k)$ is decided by $\mathcal{P}(n)$ with two additional linear constraints.		 
\end{lemma}
\begin{proof}
	We show that, given $X\in \mathcal{P}(n)$, if $X$ satisfies the following
	two linear constraints: $tr(X)=k,\: sum(X)=k^{2}$, then $X\in \mathcal{P}(n,k)$.
	
	Since $X\in \mathcal{P}(n)$, there exits $X_{i}\in \mathcal{P}(n,i)$ and $p_{i}\geq0$,
	such that $X=\sum_{i=1}^{n}p_{i}X_{i}$ and $\sum_{i=1}^{n}p_{i}=1$. That is, $X$ is a convex combination of elements from each $\mathcal{P}(n,i)$. 
	Notice that $\forall X_{i}\in \mathcal{P}(n,i)$, we have $tr(X_{i})=i$ and
	$sum(X_{i})=i^{2}$, since any vertex of $\mathcal{P}(n,i)$ satisfies these
	constraints. Let $X \in \mathcal{P}(n,k)$, then we have:	
	\begin{align*}
	(i): & 1=\sum_{i=1}^{n}p_{i}\\
	(ii): & k=tr(X)=\sum_{i=1}^{n}p_{i}\times tr(X_{i})=\sum_{i=1}^{n}p_{i} \times i\\
	(iii): & k^{2}=sum(X)=\sum_{i=1}^{n}p_{i} \times sum(X_{i})=\sum_{i=1}^{n}p_{i} \times i^{2}
	\end{align*}

	By the Cauchy-Schwarz inequality, we have $(\sum_{i=1}^{n}p_{i})(\sum_{i=1}^{n}p_{i}i^{2})\geq(\sum_{i=1}^{n}p_{i}i)^{2}$.
	Plugging in the above three equations into the Cauchy-Schwarz inequality yields that the equality holds.
	The condition of equality for the Cauchy\textendash{}Schwarz inequality
	is that $p_{i}i^{2}/p_{i}$ is a constant for all $i$, such that
	$p_{i}\not=0$. This shows that there is only one non-zero among $p_{i}$'s. That is $p_{k}=1$. Therefore, $X\in \mathcal{P}(n,k)$.
\end{proof}
\emph{Remark:} we note that \cite{Corpoly} defines correlation polytope in a more general fashion and our definition of $\mathcal{P}(n)$ is in fact an important special case of the definition of correlation polytope in \cite{Corpoly}, which is called the full correlation polytope. Nevertheless, this definition is sufficient for our model.  \cite{Corpoly} proved that membership check for polytope $\mathcal{P}(n)$ is NP-complete. Lemma \ref{lem: polytope property} basically conveys that optimizing over polytope $\mathcal{P}(n,k)$ is no harder than optimizing over  $\mathcal{P}(n)$.  Nevertheless, Lemma~\ref{lem:polyOracle} shows that optimizing over $\mathcal{P}(n,k)$ is still NP-hard.
\end{document}